\documentclass{article}

\usepackage{amsmath, amssymb, amsfonts, amsthm}
\usepackage{graphicx}
\usepackage{hyperref}
\usepackage{optidef}

\usepackage{xcolor}

\usepackage{enumerate}

\usepackage{subcaption}

\usepackage{authblk}

\begin{document}
\newcommand{\norm}[2]{\lVert{#2}\rVert_{#1}} 
\newcommand{\diagm}[1]{\mathrm{diag}\left(#1\right)} 
\newcommand{\normweighted}[3]{\lVert{#2}\rVert_{#1,\diagm{#3}}} 
\newcommand{\inv}{^{-1}} 
\newcommand{\R}{\mathbb{R}} 
\newcommand{\Rpos}{\mathbb{R}_{>0}} 
\newcommand{\Rnonneg}{\mathbb{R}_{\geq 0}} 
\newcommand{\sat}[1]{\mathrm{sat}\left({#1}\right)} 
\newcommand{\dz}[1]{\mathrm{dz}\left({#1}\right)} 
\newcommand{\sati}[1]{\mathrm{sat}_i\left({#1}\right)} 
\newcommand{\dzi}[1]{\mathrm{dz}_i\left({#1}\right)} 
\newcommand{\sign}[1]{\mathrm{sign}({#1})} 
\newcommand{\signi}[1]{\mathrm{sign}_i\left({#1}\right)} 
\newcommand{\ones}{\mathbf{1}} 
\newcommand{\pospart}[1]{\left[ #1 \right]_+}
\newcommand{\negpart}[1]{\left[ #1 \right]_-}

\newcommand{\eqrefbetween}[2]{\eqref{#1}--\eqref{#2}}

\newcommand{\todo}[1]{\textcolor{orange}{#1}}

\newcommand{\tsat}[1]{\Tilde{\mathrm{sat}}\left({#1}\right)} 
\newcommand{\tdz}[1]{\Tilde{\mathrm{dz}}\left({#1}\right)} 
\newcommand{\tsati}[1]{\Tilde{\mathrm{sat}}_i\left({#1}\right)} 
\newcommand{\tdzi}[1]{\Tilde{\mathrm{dz}}_i\left({#1}\right)} 
\newcommand*{\tz}{\Tilde{\zeta}}
\newcommand*{\tu}{\Tilde{u}}

\newcommand{\vu}{\overline{v}}
\newcommand{\vl}{\underline{v}}

\newcommand{\uu}{\overline{u}}
\newcommand{\ul}{\underline{u}}
\newcommand{\kp}{k^{\mathrm{P}}} \newcommand{\kpi}{k^{\mathrm{P}}_i} \newcommand{\KP}{{K^\mathrm{P}}}
\newcommand{\ki}{k^{\mathrm{I}}} \newcommand{\kii}{k^{\mathrm{I}}_i} \newcommand{\KI}{{K^\mathrm{I}}}
\newcommand{\ka}{k^{\mathrm{A}}} \newcommand{\kai}{k^{\mathrm{A}}_i} \newcommand{\KA}{{K^\mathrm{A}}}
\newcommand*{\kc}{k^\mathrm{C}}

\newtheorem{problem}{Problem}
\newtheorem{theorem}{Theorem}
\newtheorem{lemma}{Lemma}
\newtheorem{remark}{Remark}
\newtheorem{definition}{Definition}
\newtheorem{assumption}{Assumption}
\newtheorem{conjecture}{Conjecture}
\newtheorem{proposition}{Proposition}

\title{On PI-control in Capacity-Limited Networks}
\date{}

\author[1]{Felix Agner\thanks{Corresponding author: felix.agner@control.lth.se}}
\author[1]{Anders Rantzer}

\affil[1]{Department of Automatic Control, Lund University, Sweden}

\maketitle

\begin{abstract}
This paper concerns control of a class of systems where multiple dynamically stable agents share a nonlinear and bounded control-interconnection. The agents are subject to a disturbance which is too large to reject with the available control action, making it impossible to stabilize all agents in their desired states. In this nonlinear setting, we consider two different anti-windup equipped proportional-integral control strategies and analyze their properties. We show that a fully decentralized strategy will globally, asymptotically stabilize a unique equilibrium. This equilibrium also minimizes a weighted sum of the tracking errors. We also consider a light addition to the fully decentralized strategy, where rank-1 coordination between the agents is introduced via the anti-windup action. We show that any equilibrium to this closed-loop system minimizes the maximum tracking error for any agent. A remarkable property of these results is that they rely on extremely few assumptions on the interconnection between the agents. Finally we illustrate how the considered model can be applied in a district heating setting, and demonstrate the two considered controllers in a simulation.
\end{abstract}

\section{Introduction}
In this paper we consider control systems where a large number of interconnected agents share a limited resource, with the goal of utilizing this resource in an optimal fashion. This form of problem arises in many real-world domains: Communication networks \cite{kelly_rate_1998,low_optimization_1999,low_internet_2002,kelly_fairness_2003}, power systems \cite{opf_pursuit_DallAnese,molzahn_survey_2017,ortmann_deployment_2023}, building cooling systems \cite{kallesoe2019_hydronic_cooling,kallesoe2020_hydronic_cooling}, district heating and cooling networks \cite{AGNER2022100067}, and distributed camera systems \cite{martins_control-based_2020,martins_dynamic_2021}.

From a control-theoretic perspective, this family of problems poses several interesting challenges. Firstly, the multi-agent setting calls for control solutions which are distributed or decentralized to maintain scalability in large networks. Secondly, the nonlinearity imposed by the resource constraint means that a fully linear systems perspective will be insufficient. Thirdly, it is often the case that a detailed system model is difficult to obtain. Hence an explicit system model may be unavailable for control design. Finally, due to the constrained resource of the system, it is often impossible to drive the system to a preferable state for all agents. Hence it becomes interesting to analyze the optimality of any equilibrium stabilized by the closed-loop system.

Early works in this direction concerned with congestion in communication networks \cite{kelly_rate_1998,low_optimization_1999}. Since then, a larger body of literature has grown. An often-considered approach is to design the closed loop system to act as a gradient-descent algorithm \cite{krishnamoorthy_real-time_2022,hauswirth_optimization_2024}, in order to ensure optimality of the resulting equilibrium. This approach faces the challenge that the gradient of the steady-state map from input to equilibrium states needs to be known. Additionally, the resulting controller inherits the structure of this gradient, which may in general be dense. While works have been published in the directions of data-driven estimation of this gradient \cite{data-driven-feedback-optimization}, there are still major challenges in multi-agent and continuous-time settings. For specific problem-instances, asymptotically optimal control solutions with structural sparsity have been shown. For network flow-control, distributed solutions have been found which yield asymptotic optimality \cite{BAUSO_asymptotical_optimality,fair_network_flow_control}. For agents connected via a saturated, linear map, where the linear part corresponds to an M-matrix, fully decentralized and rank-1 coordinated control has been considered \cite{AGNER2024101049_decentralized_pi,AGNER2023_antiwindup_coordination_strategy}. These two works consider anti-windup-equipped proportional-integral control. Anti-windup has a long history of use in dynamic controllers for plants with input saturations, typically with the purpose of ensuring that the behavior of the controller in the saturated region does not drastically differ from the unsaturated behavior \cite{anti-windup-tutorial}. However, recent works have also shown that anti-windup has a useful application in real-time optimization \cite{aw_implementation_for_optimization} as it holds an interpretation of projection onto the feasibility set of the system. In \cite{martins_control-based_2020,martins_dynamic_2021}, an anti-windup-based controller is heuristically proposed and used to coordinate the allocation of a limited volume of disk space within a distributed camera system, informing the different cameras in the network of the current resource availability and thus improving the resource usage.

In this paper we present the following contributions. We study an extension of the model of capacity-constrained systems considered in \cite{AGNER2024101049_decentralized_pi,AGNER2023_antiwindup_coordination_strategy} to a fully nonlinear setting. We show that this extension to the nonlinear domain is crucial for modeling real world systems by explicitly demonstrating how the model can capture a district heating network. For the considered model, we consider the same two forms of controllers based on anti-windup-equipped PI control as considered in \cite{AGNER2024101049_decentralized_pi,AGNER2023_antiwindup_coordination_strategy}. Firstly a fully decentralized control structure, and secondly a structure which introduces light rank-1 coordination between the agents. We show that the results presented in \cite{AGNER2024101049_decentralized_pi,AGNER2023_antiwindup_coordination_strategy} still hold in a fully nonlinear setting. In particular, the fully decentralized controller globally, asymptotically stabilizes the system, and both of the considered controllers admit closed-loop equilibria which are optimal in the following ways: The fully decentralized controller minimizes a cost on the form $\sum a_i \nu_i |x_i|$, and the coordinated controller minimizes the largest control error $\norm{\infty}{x}$. 

We formally introduce the considered plant and problem formulation in Section \ref{sec:prolem formulation}. We present the two considered control strategies, along with their associated theoretical results on stability and optimality in Section \ref{sec:main results}. We demonstrate the applicability of the considered controllers in a motivating example based on district heating in Section \ref{sec:motivating example}, along with a simulation. In Section \ref{sec:proofs} we prove the main results of the paper and we finally conclude the paper in Section \ref{sec:conclusion}.

\subsection{Notation}
For a vector $v \in \R^n$, we denote $v_i$ to be element $i$ of $v$. We denote $\diagm{v}$ to be a diagonal matrix with the elements of the vector $v$ along its' diagonal. We denote $\Rnonneg$ ($\Rpos$) to be the set of non-negative (positive) numbers. If $v$ and $u$ are two vectors in $\R^n$, we say that $v \geq u$ ($v > u$) if $v-u \in \Rnonneg^n$ ($v-u \in \Rpos^n$). We denote $\pospart{v}$ to be the element-wise non-negative parts of the elements of $v$, such that if $u = \pospart{v}$, then $u_i = \max(v_i, 0)$. Conversely, $\negpart{v} = v - \pospart{v}$. We define the saturation function as $\sat{u}_i = \max\left(\underline{l}_i, \, \mathrm{min} \left( \overline{l}_i, \, u_i \right) \right)$. $\sat{\cdot}$ maps from $\R^n$ to a set $\mathcal{S} = \{ v \in \R^n \, | \, \underline{l}_i \leq v_i \leq \overline{l}_i, \, \forall i = 1, \dots, n \}$ defined by the bounds $\overline{l}$, $\underline{l}$. When $\sat{\cdot}$ is applied to an element of a vector, e.g., $\sat{u_i}$, the bounds $\overline{l_i}, \underline{l_i}$ are implicitly used. We define the dead-zone nonlinearity $\dz{u} = u - \sat{u}$. We denote the sign-function $\sign{x} = x / |x|$ when $x \neq 0$ and $\sign{x} = 0$ for $x=0$. When we apply $\sign{\cdot}$ to a vector, the operation is performed element-wise. 
For a vector $x \in \R^n$ we use the $l_1$-and-$l_\infty$-norms $\norm{1}{x} = \sum_{i=1}^n \left| x_i \right|$ and $\norm{\infty}{x} = \max_i \left| x_i \right|$ respectively. 

\section{Problem Formulation} \label{sec:prolem formulation}
In this section we first introduce the considered plant, and subsequently the associated control problem.

\subsection{Plant Description}
We consider the control of multi-agent systems where the dynamics of agent $i \in 1, \dots, n$ can be described by the following dynamics.
\begin{equation}
    \dot{x}_i = -a_i x_i + b_i\left(\sat{u}\right) + w_i \label{eq:agent dynamics}
\end{equation}
Here $x_i$ denotes the scalar state of agent $i$ which should be maintained close to 0. $a_i \in \R_{>0}$ models a stable internal behavior of agent $i$. $w_i \in \R$ is a disturbance acting on agent $i$, assumed to be constant.  $b_i$ is the $i$'th component of a nonlinear interconnection $b : \mathcal{S} \to \mathcal{B}$ between the agents. Here $\mathcal{S}$ is the range of the saturation function. We consider the case where $b$ is not explicitly known and hence cannot be used in control design and actuation. However, we assume that $b$ holds certain exploitable properties:
\begin{assumption} (Input-output properties of $b$) \label{ass:assumptions on b}
    $b : \mathcal{S} \to \mathcal{B}$ is a continuous function. There exists $\eta \in \R^n_{>0}$ such that for any pair $\vu$, $\vl$ $\in \mathcal{S}$ where $\vu \geq \vl$ and $\vu \neq \vl$,
    \begin{enumerate}[(i)]
        \item $b_i(\vu) - b_i(\vl) < 0$ if $\vu_i = \vl_i$, and
        \item $\eta^\top\left(b\left(\vu\right) - b\left(\vl\right)\right) > 0$. 
    \end{enumerate}
\end{assumption}
Assumption $(i)$ encodes competition between the agents: if other agents \textit{increase} their control action while agent $i$ maintains their control input ( $\vu_i = \vl_i$ and $\vu \geq \vl$), the resource granted to agent $i$ decreases ($b_i(\vu) - b_i(\vl) < 0$). Assumption $(ii)$ encodes that if all agents increase their system input ($\vu \geq \vl$), the output of the system increases ($\eta^\top\left(b\left(\vu\right) - b\left(\vl\right)\right) > 0$). This increase concerns a weighted output, governed by a weight $\eta$. If $b$ satisfies $(ii)$ for many different vectors $\eta > 0$, then the results of this paper hold for any such choice of $\eta$.

\begin{remark}
    Assumption \ref{ass:assumptions on b} is satisfied in the linear case when $b(v) = Bv$ and $B\in \R^{n \times n}$ is an M-matrix. $(i)$ then corresponds to the non-positivity of $B$'s off-diagonal elements. $B$ also has a positive left eigenvector $\eta$ with associated positive eigenvalue $\lambda$ such that $\eta^\top B = \lambda \eta^\top$, which implies $(ii)$. This is the case investigated in \cite{AGNER2024101049_decentralized_pi,AGNER2023_antiwindup_coordination_strategy}. We refer to \cite[pp. 113-115]{horn_johnson_1991} for a more detailed definition of M-matrices and a list of their properties.
\end{remark}

\begin{remark}
Note that $\mathcal{S}$ is an $n$-dimensional box and thus compact. As $b$ is continuous, $\mathcal{B}$ is therefore also compact due to the extreme value theorem.    
\end{remark}

\subsection{Problem Description}
In an ideal scenario, a controller should drive the system \eqref{eq:agent dynamics} to the origin ($x=0$), which means that there are no control errors. This is unforunately not always possible. The dynamics \eqref{eq:agent dynamics} dictate that any equilibrium state-input pair ($x^0, u^0$) yielding $\dot{x} = 0$ must satisfy $a_i x_i^0 = b_i(\sat{u^0}) + w_i$ for all $i = 1,\dots,n$. But when the disturbance $w$ is large, we may find that $-w \notin \mathcal{B}$ as the image $\mathcal{B}$ of $b$ is compact. Thus it becomes impossible to stabilize the origin. In this scenario, our aim is to design controllers which stabilize an equilibrium close to the origin, where we will consider two such notions of "close". The multi-agent setting also provides the complication that the controllers should require little to no communication. Furthermore, we have no explicit model of $b$, and can therefore not use it for control design or actuation.

\section{Considered Controllers and Main Results} \label{sec:main results}
In this section we will define two proportional-integral control strategies. In the unsaturated region, both controllers are equivalent and fully decentralized. In the saturated region they are equipped with different anti-windup compensators. One of these anti-windup compensators is fully \textit{decentralized} and the other is \textit{coordinating} using rank-1 communication. We will show how the closed-loop equilibria of these two strategies minimize the distance to the origin by two different metrics.

\subsection{Decentralized Control}

The first control strategy we investigate is also the simplest, namely the fully decentralized strategy. Each agent $i = 1, \dots, n$, is equipped with an integral error $z_i$, proportional and integral gains $\kpi \in \Rpos$ and $\kii \in \Rpos$, and an anti-windup gain $\kai \in \Rpos$. Their closed loop system is therefore described by
\begin{subequations}
\begin{align}
    \dot{x}_i &= -a_i x_i + b_i(\sat{u}) + w_i \label{eq:decentralized x}\\
    \dot{z}_i &= x_i + \kai \dzi{u} \label{eq:decentralized z}\\
    u_i &= -\kpi x_i - \kii z_i. \label{eq:decentralized u}
\end{align}
\label{eq:decentralized closed loop}
\end{subequations}

We assume that the controller gains of each agent are tuned according to the following rule.
\begin{assumption} \label{ass:decentralized controller tuning}
    For all agents $i = 1,\dots,n$, it holds that $\kpi a_i > \kii$ (the proportional gain dominates the integral gain) and $\kpi \kai < 1$ (the proportional gain and the anti-windup gain are limited). 
\end{assumption}
Note that this control strategy is fully decentralized not only in terms of actuation, but also in terms of Assumption \ref{ass:decentralized controller tuning}. The controller tuning also requires no explicit model of the interconnection $b$. For this closed-loop system, we present the following qualities, which we will later prove in Section \ref{sec:proofs}.

\begin{theorem}[Global asymptotic stability] \label{thm: decentralized stability}
    Let Assumptions \ref{ass:assumptions on b} and \ref{ass:decentralized controller tuning} hold. Then the closed loop system \eqref{eq:decentralized closed loop} formed by the decentralized controller has a unique, globally asymptotically stable equilibrium.
\end{theorem}

This theorem is proven in Section \ref{sec:stability proofs}. By an equilibrium in this context, we mean a pair $(x^0, z^0)$ with associated control input $u^0$ which solves \eqref{eq:decentralized closed loop} with $\dot{x} = \dot{z} = 0$. We can show that this equilibrium is optimal in the following sense.

\begin{theorem}[Equilibrium optimality] \label{thm: decentralized optimality}
    Let Assumptions \ref{ass:assumptions on b} and \ref{ass:decentralized controller tuning} hold, and recall the vector $\eta$ from Assumption \ref{ass:assumptions on b}. Let $(x^0, u^0)$ be the equilibrium state-input pair stabilized by the decentralized controller \eqref{eq:decentralized closed loop}. Consider any other pair $x^\dagger \in \R^n$, $u^\dagger \in \R^n$ which forms an equilibrium for the open-loop system, i.e., which solves \eqref{eq:agent dynamics} with $\dot{x} = 0$. If $\sat{u^\dagger} \neq \sat{u^0}$, then
    \begin{equation}
        \sum_{i=1}^n \eta_i a_i | x^0_i | < \sum_{i=1}^n \eta_i a_i | x^\dagger_i |.
    \end{equation}
\end{theorem}

This theorem is proven in Section \ref{sec:optimality proofs}. This optimality guarantee is given for the objective $\sum_{i=1}^n \eta_i a_i | x^0_i | $, characterized by $a$ and $\eta$. Hence this result does not yield a control design method for minimizing general costs on the form $\norm{1}{Wx}$, where $W$ is an arbitrary weight. Rather, it highlights that for an interesting class of problems, this fully decentralized controller which is designed without explicit parameterization of $b$ can still provide a notion of optimality. In fact, if Assumption \ref{ass:assumptions on b} is satisfied for a whole set of vectors $\eta$, the optimality of Theorem \ref{thm: decentralized optimality} holds for all such vectors. 

\subsection{Coordinating Control}

The second control strategy we consider introduces a coordinating anti-windup signal. The proposed closed-loop system is given by
\begin{subequations}
\begin{align}
    \dot{x}_i &= -a_i x_i + b_i(\sat{u}) + w_i \label{eq:coordinated x}\\
    \dot{z}_i &= x_i + \kc \ones^\top \dz{u} \label{eq:coordinated z}\\
    u_i &= -\kpi x_i - \kii z_i. \label{eq:coordinated u}
\end{align}
\label{eq:coordinated closed loop}
\end{subequations}
The only difference from the decentralized strategy \eqref{eq:decentralized closed loop} is the coordinating anti-windup term $ \kc \ones^\top \dz{u}$, where $\kc \in \Rpos$ is an anti-windup gain. This coordinating controller is also fully decentralized in the unsaturated domain $\dz{u} = 0$. When saturation occurs, the communication is rank-1, hence it can be implemented simply though one shared point of communication, or via scalable consensus-protocols \cite{consensus_murray}. The coordinating term $ \kc \ones^\top \dz{u}$ heuristically embeds the following idea: If the current disturbance on the system is large, and an agent requires more control action than the saturation allows ($\dz{u_k}$ large for some $k \in 1,\dots,n$), then this will enter into the coordinating term $ \kc \ones^\top \dz{u}$ and make all other agents reduce their control action, freeing more of the shared resource. 

We assume that the coordinating controller is designed according to the following rules.
\begin{assumption} \label{ass:coordinating controller tuning}
    For all agents $i = 1,\dots,n$, it holds that $a_i \kpi = (1+\alpha) \kii$, where $\alpha \in \Rpos$ is a tuning gain known to all agents. Additionally, the anti-windup gain $\kc \in \Rpos$ is chosen sufficiently small, such that $\frac{\kc}{2} \ones^\top \kp \leq 1$. 
\end{assumption}

Equation \eqref{eq:coordinated z} imposes the equilibrium condition $-x^0 = \kc \ones \ones^\top \dz{u^0}$, i.e., $x^0$ is parallel to $\ones$. This means that in any closed-loop equilibrium, the imposed control error is shared equally between all agents. In a sense, this means that the resource is being shared in a fair fashion between the agents. In fact, the imposed equilibrium will be optimally fair in the following sense.

\begin{theorem}[Equilibrium optimality] \label{thm: coordinating optimality}
    Let Assumptions \ref{ass:assumptions on b} and \ref{ass:coordinating controller tuning} hold. Assume that $(x^0, u^0)$ is an equilibrium stabilized by the coordinating controller \eqref{eq:coordinated closed loop}. Consider any other pair $x^\dagger \in \R^n$, $u^\dagger \in \R^n$ forming an equilibrium for the open-loop system, i.e., they solve \eqref{eq:agent dynamics} with $\dot{x} = 0$. If $\sat{u^\dagger} \neq \sat{u^0}$ then
    \begin{equation}
        \norm{\infty}{x^0} < \norm{\infty}{x^\dagger}. 
    \end{equation}
\end{theorem}

This theorem is proven in Section \ref{sec:optimality proofs}. While this is a strong result, it is only interesting if two implicit assumptions are satisfied: That such an equilibrium exists, and that it is globally (or at least locally) asymptotically stable. However, this is not always the case. As stated previously, \eqref{eq:coordinated z} imposes that any equilibrium $x^0$ is parallel to $\ones$. At the same time, \eqref{eq:coordinated x} imposes that any equilibrium satisfies $x^0 = A\inv \left(b(\sat{u})+w\right)$ where $A = \diagm{a}$. As $b$ is bounded, these two relations can only hold if $A\inv w$ is approximately parallel to $\ones$, i.e., if $w$ has a similar effect on each agent. An exact characterization of such a condition on $w$ is outside the scope of this work. We refer to \cite{AGNER2023_antiwindup_coordination_strategy} for the case where $b$ is linear. Furthermore, even when an equilibrium exists, it is non-trivial to show that the equilibrium will be stable. Such an exercise is outside the scope of most regular stability analysis for saturating systems, where it is often assumed that the stabilized equilibrium lies in the unsaturated region. We can however show the following result, which applies when the the disturbance is small enough to be rejected. This theorem is proven in Section \ref{sec:stability proofs}.

\begin{theorem}[Global asymptotic stability] \label{thm: coordinating stability}
    Let Assumptions \ref{ass:assumptions on b} and \ref{ass:coordinating controller tuning} hold. Additionally, assume that $b(\overline{l}) + w > 0$ and $b(\underline{l}) + w < 0$. Then the closed loop system \eqref{eq:coordinated closed loop} formed by the coordinating controller has a unique, globally asymptotically stable equilibrium.
\end{theorem}

\section{Motivating Example - District Heating} \label{sec:motivating example}

\begin{figure*}
    \centering
    \includegraphics[width=\linewidth]{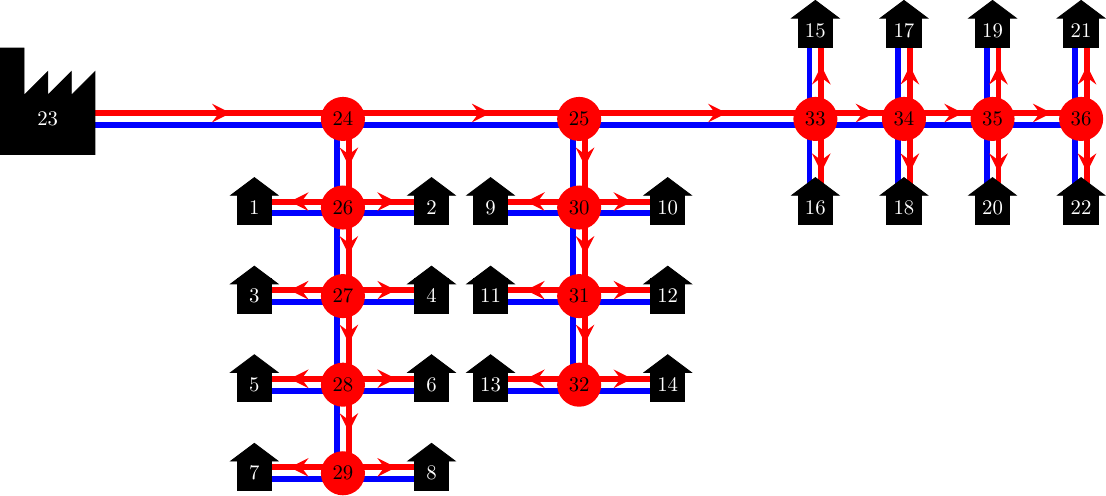}
    \caption{A district heating network. A heating plant (node 23) heats up water and pumps it out to consumers (nodes 1-22) through the supply-side network (red edges). The water subsequently returns through the return-side network (blue edges).}
    \label{fig:dhn schematic}
\end{figure*}

To illustrate the usefulness of the theoretical results, we consider district heating networks as a motivating example. Figure \ref{fig:dhn schematic} shows a schematic example of such a system. Typically in existing networks of traditional design, one or a few large heating plants produce hot water which is pumped out to consumers via a network of pipelines (red edges in Figure \ref{fig:dhn schematic}). Each consumer is equipped with a valve to regulate the amount of hot water they receive. This water runs through a heat exchanger in which heat is transferred to the internal heating system of the building. The water subsequently returns through another network of pipes (blue edges in Figure \ref{fig:dhn schematic}) which is symmetric to the supply-side network. The primary aim in the network is to supply enough hot water to the consumers, such that they can maintain comfort temperatures within their buildings. We consider a simple dynamical model for the temperature $T_i$ in each building $i=1,\dots,n$, which consumer $i$ would like to maintain at a reference temperature $T^r_i$. Hence the tracking error is $x_i = T_i - T^r_i$. The dynamics guiding the tracking error $x_i$ is thus given by
\begin{equation}
    c_i \dot{x}_i = -\hat{a}_i\left( T^r_i + x_i - T_o \right) + c_{p,w}\rho_w \delta_i q_i(v).
    \label{eq:building dynamics}
\end{equation}
Here $c_i$ is the heat capacity of building $i$. $\hat{a}_i$ is the thermal conductance and $T_o$ is the outdoor temperature, acting as a disturbance. Hence the first term $-\hat{a}_i\left( T^r_i + x_i - T_o \right)$ corresponds to diffusion of heat between the interior and exterior of the building. $c_{p,w}$ and $\rho_w$ are the specific heat capacity and density of water respectively (both assumed constant) and $q_i$ is the volume flow rate going through the building heat exchanger. $\delta_i$ is the difference in supply-and-return temperature before and after the heat exchanger. Hence the second term $c_{p,w}\rho_w \delta_i q_i(v)$ corresponds to the heat provided to the building through the heat exchanger. The volume flow rate $q_i$ is regulated by the valve positions $v = \sat{u}$. Note that the flow rate $q_i$ provided to consumer $i$ is influenced by the valve positions $v$ of all consumers in the network, not only $v_i$. This is because the pressure distribution in the network is affected by all of the flow rates in the network. Herein lies the main connection to our theoretical results. As the central pump is limited in its maximum capacity, and the valves themselves are saturated, the volume flow rate and hence the heat that can be supplied to the consumers is limited. Hence when the disturbance $T_o$ is sufficiently low, the available capacity becomes insufficient.

To complete the connection between the temperature model \eqref{eq:building dynamics} and the agent dynamics \eqref{eq:agent dynamics} as we have considered in this paper, we can identify $a_i = \frac{\hat{a}_i}{c_i}$, $b_i(\sat{u}) = \frac{c_{p,w}\rho_w \delta_i}{c_i}q_i(\sat{u})$ and $w_i = \frac{\hat{a}_i}{c_i}(T_o = T^r_i)$. Secondly, we make the following simplifying assumption.
\begin{assumption}
    The delta temperatures $\delta_i$, $c_{p,w}$ and $\rho_w$ are all constant.
\end{assumption}
In practice this assumption will not hold exactly. The delta temperature changes slightly with several factors, such as the supply-temperature in the network, the activity on the secondary side of the heat exchanger (i.e., the side facing the consumers internal heating system). There are also slight temperature-dependent variations in the density of the water. In general however and over shorter time-spans, these variations are much smaller than the variations in volume flow rate. This assumption means that the final verification to make is that $q$ satisfies Assumption \ref{ass:assumptions on b}. Under the assumption that we use common static models for the valves and pipes in the network such as in \cite{de_persis_output_2014,AGNER2022100067,jeeninga_existence_2023}, that the network is tree-structured, and at the pump at the root of the tree operates at constant capacity, we can show the following. Hence $q$ satisfies Assumption \ref{ass:assumptions on b}.
\begin{proposition}
    Given two sets of valve positions $\vu \geq \vl$,
    \begin{enumerate}[(i)]
        \item $q_i(\vu) - q_i(\vl) \leq 0$ if $\vu_i = \vl_i$, and
        \item $\ones^\top \left( q(\vu) - q(\vl) \right) > 0$ if $\vu \neq \vl$.
    \end{enumerate}
\end{proposition}
We omit the proof of this proposition, as it demands a technical description of district heating hydraulic models. However, we can motivate the proposition in the following way. If all agents incrementally open their valves ($\vu \geq \vl$), this reduces the resistance in the system, which means that the total throughput increases ($\ones^\top q(\vu) > \ones^\top q(\vl)$), i.e., $(ii)$. However, as the total throughput increases, so do pressure losses in the pipelines. Hence, if one valve $i$ is unchanged ($\vu_i = \vl_i$), the flow rate through the valve will decrease due to reduced differential pressure ($q_i(\vu) \leq q_i(\vl) $), i.e., $(i)$.

\subsection{Numerical Example}

To investigate the effect of the considered control strategies in a district heating setting, we perform a simulation of a small district heating network. The network is structured as in Figure \ref{fig:dhn schematic}, and each building is subject to the dynamics given in \eqref{eq:building dynamics}. For simplicity, we consider a homogeneous building stock with $c_i = 2.0 [\text{kWh/K}]$, $a_i = 1.2 [\text{kW/K}]$ and $\delta_i = 50.0 [\text{K}]$ $\forall i = 1,\dots,n$. We have $c_w = 1.16 \cdot 10^{-3} [\text{kWh/kgK}]$ and $\rho_w = 10^{3} [\text{kg/m}^3]$. While we omit a detailed description of district heating hydraulics here, we use the same type of graph-based modeling as is used in \cite{agner_aalborg}. We assume that the heating plant supplies a differential pressure of $0.6 \cdot 10^6 [\text{Pa}]$. The difference between the input and output of each pipe $e$ is given by $ \Delta p_e= s_e |q_e| q_e$ where $s_e$ corresponds to a hydraulic resistance. We use $s_e = 0.9 [\text{Pa / (m$^3$/h)}^2]$ for the long edge connecting nodes 23 and 24. We use $s_e = 0.25 [\text{Pa / (m$^3$/h)}^2]$ for the edges connecting nodes 24, 25, 26, 30 and 33. We use $s_e = 0.05 [\text{Pa / (m$^3$/h)}^2]$ for the pipes connecting nodes 26-27-28-29, nodes 30-31-32 and nodes 33-34-35-36. Finally we use $s_e = 2.5 [\text{Pa / (m$^3$/h)}^2]$ for the connection to each consumer. The pressure difference between supply-and-return-side for consumer $i$ is modeled as $\Delta p_i(q_i, v_i) = \left(5 + \frac{30}{(v_i + 1.001)^2}\right)q_i^2$. Here $v_i = \sat{u_i}$ is limited in the $v_i \in \left[-1, 1\right]$. The component $5q_i^2$ corresponds to inactive components of the consumer substation, i.e., the heat exchanger and internal piping. The remaining component $30q_i^2/(v_i+1.001)^2$ corresponds to the pressure loss over the valve. 

We subject the buildings to an outdoor temperatures disturbance $T_o$ as seen in Figure \ref{fig:outdoor temperature}, acting equally on all buildings. The temperature drops critically to below $-25^\circ$C around 50 hours into the simulation. The temperature is based on temperature data from Gävle, Sweden on January 18th-21st, 2024. The data is collected from the Swedish Meteorological and Hydrological Institute.
\begin{figure}
    \centering
    \includegraphics[width=.7\linewidth]{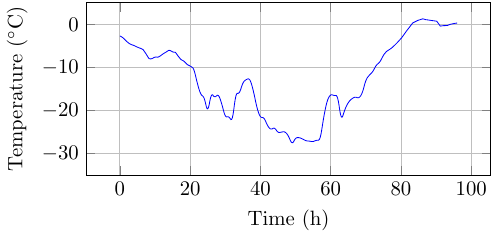}
    \caption{Outdoor temperature $T_o(t)$ used in simulation.}
    \label{fig:outdoor temperature}
\end{figure}

We consider the two control policies analyzed in this paper, namely the \textit{decentralized} and \textit{coordinating} control policies. We employ identical controllers for each consumer with $\kpi = 1.0$, $\kii = 1.0$, $\kai = 1.0$ for all $i = 1,\dots, n$ in the decentralized case, and  $\kpi = 1.0$, $\kii = 1.0$, $\kai = 1.0$, $\kc = 0.5$ in the coordinating case. As a benchmark, we compare these strategies to optimal counter-parts. In these benchmarks, the volume flow rate $q(t)$ is distributed optimally in each instance of the simulation as the solution to the problem
\begin{mini!} 
    {x,q}{J(x)}
    {\label{eq:optimal equilibrium problem}}{\label{eq:optimal cost}}
    \addConstraint{\eqref{eq:building dynamics} \text{ with } \dot{x}_i = 0 \text{ for } i = 1,\dots,n} \label{eq:temperature equilibrium}
    \addConstraint{q \in \mathcal{Q}} \label{eq:q in Q}
\end{mini!}
where we use $J(x) = \norm{1}{x}$ and $J(x) = \norm{\infty}{x}$ respectively. This problem corresponds to calculating a flow rate $q$ which is feasible within the hydraulic constraints of the network (i.e., \eqref{eq:temperature equilibrium}) which generates an equilibrium (i.e., \eqref{eq:temperature equilibrium}) which minimizes $J(x)$. To see how this optimization problem can be cast as a convex problem, we refer to \cite{AGNER2022100067} in which it is shown that $\mathcal{Q}$ is convex.

We use the \texttt{DifferentialEquations} toolbox \cite{rackauckas2017differentialequations} in \texttt{Julia} to simulate the system, utilizing the \texttt{FBDF} solver. We use the \texttt{NonlinearSolve} \cite{pal2024nonlinearsolve} toolbox to calculate $q$ as a function of the valve positions $v = \sat{u}$. We use the \texttt{Convex} toolbox \cite{convexjl} with the \texttt{Mosek} optimizer to find the optimal trajectories for the benchmark comparisons.

\begin{figure*}
    \begin{subfigure}[b]{0.5\textwidth}
         \centering
         \includegraphics[width=\linewidth]{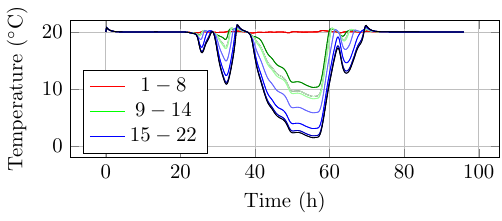}
         \caption{Decentralized PI control.}
         \label{fig:result decentralized pi}
     \end{subfigure}
     \hfill
     \begin{subfigure}[b]{0.5\textwidth}
         \centering
         \includegraphics[width=\linewidth]{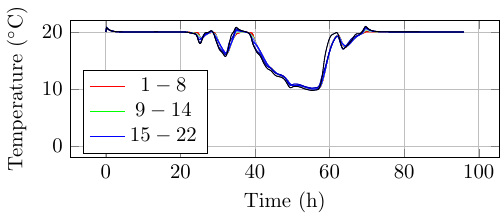}
         \caption{Coordinated PI control.}
         \label{fig:result coordinated pi}
     \end{subfigure}

     \begin{subfigure}[b]{0.5\textwidth}
         \centering
         \includegraphics[width=\linewidth]{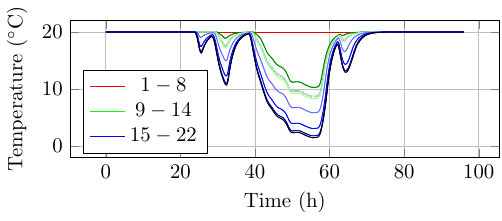}
         \caption{Optimal equilibrium input with respect to $\norm{1}{x}$.}
         \label{fig:result optimal 1}
     \end{subfigure}
     \hfill
     \begin{subfigure}[b]{0.5\textwidth}
         \centering
         \includegraphics[width=\linewidth]{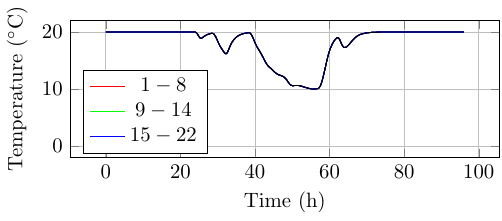}
         \caption{Optimal equilibrium input with respect to $\norm{\infty}{x}$.}
         \label{fig:result optimal inf}
     \end{subfigure}
     \caption{Resulting indoor temperatures. The three clusters of buildings, i.e., nodes 1-8, 9-14 and 15-22 are colored in red, green and blue respectively. The darker shades of each color shows the buildings further down each line.}
     \label{fig:results}
\end{figure*}

The results of the four simulations are seen in Figure \ref{fig:results}. Under all four policies, the temperatures in the buildings drop at several points during the simulation, and most significantly starting after around 40 hours. This is because of the extremely cold temperature at this time, for which the available pumping capacity is insufficient. We can first compare the results of using the decentralized strategy to the results of using the optimal equilibrium input with regards to $\norm{1}{x}$, as seen in Figures \ref{fig:result decentralized pi} and \ref{fig:result optimal 1} respectively. We find that they are effectively the same, except for minor oscillations around the equilibrium in the PI-controller case, caused by the integral action in the controller. The same comparison can be drawn between Figures \ref{fig:result coordinated pi} and \ref{fig:result optimal inf}, showing the results of using the coordinating strategy and the optimal input with regards to $\norm{\infty}{x}$. This is in line with Theorems \ref{thm: decentralized optimality} and \ref{thm: coordinating optimality}, where we expect out controllers to track the optimal equilibrium. While the interpretation of the optimal cost in Theorem \ref{thm: decentralized optimality} is obscured by the weight $\eta$, we see in this example how it can correspond to the combined tracking error of all agents $\norm{1}{x}$.

When comparing the results of using the decentralized strategy to using the coordinating strategy, we see the following: In the fully decentralized simulation, several of the buildings far away from the heating plant drop below 5$^\circ$C, whereas buildings close to the heating facility maintain comfort temperature. On the other hand, none of the buildings drop below 10$^\circ$C. However, none of the buildings maintain comfort temperature either. Which strategy is to be preferred is debatable and perhaps situational. Arguably in the extreme scenario of this simulation, the decentralized strategy might be preferred. Consumers will be severely dissatisfied if their indoor temperatures drop by $10^\circ$C, hence it may be better to have a lower number of consumers be very dissatisfied than to have the whole building stock be moderately dissatisfied. However, if we instead consider the temperatures distribution at around 30 hours into the simulations, we see that under the decentralized case there are buildings experiencing reductions in indoor temperature by $10^\circ$C, whereas in the coordinated case the worst reduction is approximately $4^\circ$C. In this case, it is arguably preferable to coordinate, as a $4^\circ$C temperature reduction is acceptable for a shorter period of time, whereas $10^\circ$C is too extreme to be tolerated by most consumers. The exact trade-offs and results will depend on the specific system and temperature levels at hand. What is interesting is that both of these behaviors which are optimal under different perspectives are achievable with such simple control techniques.

\section{Main Proofs}\label{sec:proofs}
We will now move on to prove the main theoretical results of this paper as presented in section \ref{sec:main results}. We will prove the stability results for both proposed controllers, i.e., Theorems \ref{thm: decentralized stability} and \ref{thm: coordinating stability}, followed by the optimality of their equilibria, i.e., Theorems \ref{thm: decentralized optimality} and \ref{thm: coordinating optimality}. First however, we will introduce a few extra properties of the interconnection $b$ which are required for the subsequent proofs.

\subsection{Additional Nonlinearity Properties}
The proofs of both of the following Lemmas are found in the Appendix.

\begin{lemma}\label{lem:sum of b}
    Let $b$ satisfy Assumption \ref{ass:assumptions on b}. Consider any pair $v \in \mathcal{S}$, $\Tilde{v} \in \mathcal{S}$ where $v \neq \Tilde{v}$. Let $\mathcal{I}^\pm = \{ i \in 1,\dots,n \, | \, v_i \neq \Tilde{v}_i \}$ and let $\mathcal{I}^0 = \{ i \in 1,\dots,n \, | \, v_i = \Tilde{v}_i \}$. Then
    \begin{align}
        &\sum_{i \in \mathcal{I}^\pm}\eta_i\sign{\Tilde{v}_i - v_i} \left(b_i(\Tilde{v}) - b_i(v)\right) \nonumber\\
        &> \sum_{k \in \mathcal{I}^0}\eta_k|b_k(\Tilde{v})-b_k(v)|. \label{eq:sum of b with signs}
    \end{align}
\end{lemma}
This lemma is proven in the appendix. An interpretation of this lemma is that the individual change in output $b_i(v) - b_i(\Tilde{v})$ goes mostly along the same sign as the corresponding individual change in input $v_i - \Tilde{v}_i$. This value for all agents who have changed their inputs ($\mathcal{I}^\pm$) dominates the change in output affecting all of the agents who did not change their inputs ($\mathcal{I}^0$). We continue with the following property of $b$.

\begin{lemma}\label{lem:b inverse positive}
    Let $b$ satisfy Assumption \ref{ass:assumptions on b}. Then for any pair $\vu \in \mathcal{S}$ and $\vl \in \mathcal{S}$ where $\vu \neq \vl$, if $b(\vu) \geq b(\vl)$, then $\vu > \vl$.
\end{lemma}
This lemma is proven in the appendix. In the linear case where $b$ is an M-matrix ($b(v) = Bv$), this property can be likened with positivity of the inverse of this matrix ($B\inv > 0$ element-wise). 

\subsection{Stability Proofs} \label{sec:stability proofs}

We will prove both Theorem \ref{thm: decentralized stability} and \ref{thm: coordinating stability} using Lyapunov-based arguments. To do so, we will first introduce a change of coordinates from $(x,z)$ to $(\zeta,u)$, where $u_i = -\kpi x_i - \kii z_i$ and $\zeta_i = - \kii z_i$ for $i = 1,\dots,n$. We also introduce the matrices $P = \diagm{\kp}$, $R = \diagm{\ki}$, $C = R P\inv$ and $D = \diagm{a} - C$. Note that $P$, $R$, $C$ and $D$ are all positive definite, diagonal matrices under either Assumption \ref{ass:decentralized controller tuning} or Assumption \ref{ass:coordinating controller tuning}. The closed-loop system in these new coordinates is then given by
\begin{equation}
    \begin{bmatrix}
        \dot{\zeta} \\
        \dot{u}
    \end{bmatrix} = \begin{bmatrix}
        -C &C \\ D &-D
    \end{bmatrix}\begin{bmatrix}
        \zeta \\
        u
    \end{bmatrix} - \begin{bmatrix}
        0 \\
        P
    \end{bmatrix}b(\sat{u}) - \begin{bmatrix}
        0 \\
        P
    \end{bmatrix}w - \begin{bmatrix}
        R \\
        R
    \end{bmatrix}S \dz{u} .
    \label{eq:changed coordinates}
\end{equation}
Here $S$ denotes the anti-windup compensation, i.e., $S = \diagm{\kai}$ for the decentralized controller and $S = \kc \ones \ones^\top$ for the coordinating controller.

We will first prove stability of the decentralized controller, beginning with ensuring that there exists an equilibrium.
\begin{lemma} \label{lem:decentralized equilibrium}
    Let Assumptions \ref{ass:assumptions on b} and \ref{ass:decentralized controller tuning} hold. Then the decentralized closed loop system \eqref{eq:decentralized closed loop} has at least one equilibrium.
\end{lemma}
This can be easily proven with Brouwer's fixed-point theorem. We provide an outline of the proof in the Appendix, omitting the details for brevity. We can now prove stability of the decentralized closed-loop system.
\begin{proof}[Proof of Theorem \ref{thm: decentralized stability}]
    Under the taken assumptions, Lemma \ref{lem:decentralized equilibrium} provides that the original system \eqref{eq:decentralized closed loop}, and thus also \eqref{eq:changed coordinates}, has at least one equilibrium which we denote $(\zeta^0, u^0)$. We introduce the shifted variables $\Tilde{\zeta} = \zeta - \zeta^0$ and $\Tilde{u} = u - u^0$, and the shifted notation $\tsat{u} = \sat{u^0 + \Tilde{u}} - \sat{u^0}$, $\Tilde{b}(\tsat{\Tilde{u}}) = b(\tsat{u} + \sat{u^0}) - b(\sat{u^0})$ and $\tdz{\Tilde{u}} = \dz{u^0 + \Tilde{u}} - \dz{u^0} = \tu - \tsat{u}$. In this coordinate frame, the closed-loop dynamics are
    \begin{equation} 
        \begin{bmatrix}
            \dot{\tz} \\
            \dot{\tu}
        \end{bmatrix} = \begin{bmatrix}
            -C &C \\ D &-D
        \end{bmatrix}\begin{bmatrix}
            \tz \\
            \tu
        \end{bmatrix} - \begin{bmatrix}
            0 \\
            P
        \end{bmatrix}\Tilde{b}(\tsat{\tu}) - \begin{bmatrix}
            R \\
            R
        \end{bmatrix}S \tdz{\tu}.
        \label{eq:shifted coordinates}
    \end{equation}
    The aim is now to show that \eqref{eq:shifted coordinates} is globally, asymptotically stable with regards to the origin, which is equivalent with Theorem \ref{thm: decentralized stability}. Now recall the vector $\eta$ of Assumption \ref{ass:assumptions on b}, which we use to define the following Lyapunov function candidate.
    \begin{equation}
        V(\Tilde{\zeta}, \Tilde{u}) = \sum_{i=1}^n \frac{\eta_i d_i}{p_ic_i}|\Tilde{\zeta}_i| + \frac{\eta_i}{p_i}|\Tilde{u}_i|
    \end{equation}

    While $V$ is not strictly continuously differentiable, we note this as a technicality. The function $| \cdot |$ can be exchanged with an arbitrarily close approximation which is continuously differentiable in the origin. For this proof, we will maintain the convention that 
    \begin{equation}
        \frac{d}{dt}|x| = \sign{x} \dot{x}.
    \end{equation}
    Denote $H = \diagm{\eta}$. We then find that
    \begin{subequations}
    \begin{align}
        &\dot{V}(\tz,\tu) \nonumber\\
        =& \sign{\tz}^\top H D P\inv C\inv \dot{\Tilde{\zeta}} + \sign{\Tilde{u}}^\top H P\inv \dot{\Tilde{u}} \nonumber\\
        =& \sign{\tz}^\top H D P\inv C\inv \left( C(\tu -\tz) - RS\tdz{\tu} \right) \nonumber\\
        & + \sign{\tu}^\top H P\inv\left( D(\tz -\tu) \right. \nonumber\\
        &\left.- P\Tilde{b}(\tsat{\Tilde{u}}) - RS\tdz{\tu} \right) \nonumber\\
        =& -\left( \sign{\tu} - \sign{\tz} \right)^\top H D P\inv\left( \tu - \tz\right)\label{eq:vdot linear} \\
        & - \sign{\tz}^\top H D P\inv C\inv R S\tdz{\tu} \label{eq:vdot cross dz}\\
        & - \sign{\tu}^\top H P\inv R S \tdz{\tu} \label{eq:vdot dz}\\
        & - \sign{\tu}^\top H \Tilde{b}(\tsat{\Tilde{u}}) \label{eq:vdot b}.
    \end{align}
    \end{subequations}
    The terms \eqrefbetween{eq:vdot linear}{eq:vdot cross dz} act fully diagonally, hence we can analyze their sign contribution for each $i \in 1,\dots,n$ individually. If either $\tu_i = 0$ and $\tz_i \neq 0$ or $\tu_i \neq 0$ and $\tz_i = 0$, then clearly \eqref{eq:vdot cross dz} contributes with 0, and \eqref{eq:vdot linear} becomes strictly negative in the non-zero variable. If $\sign{\tu_i} = \sign{\tz_i}$, then \eqref{eq:vdot linear} contributes with $0$, but \eqref{eq:vdot cross dz} contributes with a negative semidefinite term, which is strictly negative if $\tdz{\tu}_i \neq 0$. If $\sign{\tu_i} = -\sign{\tz_i} \neq 0$, \eqref{eq:vdot cross dz} contributes with a positive semidefinite term $\frac{\eta_i d_i r_i s_i}{p_i c_i} | \tdz{\tu}_i |$. However, in this case \eqref{eq:vdot linear} contributes with a strictly negative term $- \frac{\eta_i d_i}{p_i} (|\tz_i| + |\tu_i|)$. This negative term dominates the positive semidefinite term because of Assumption \ref{ass:decentralized controller tuning} where the anti-windup gain is bounded, and because $|\tdz{\tu}_i| \leq |\tu_i|$. Hence the contribution of \eqrefbetween{eq:vdot linear}{eq:vdot cross dz} is negative semidefinite, and strictly negative when $\tdz{\tu} \neq 0$. The term \eqref{eq:vdot dz} is trivially negative semidefinite, and strictly negative when $\tdz{\tu} \neq 0$. Finally, let $\mathcal{I}^0 = \{ i \in 1,\dots,n \, | \, \tsat{\tu}_i = 0 \}$. Note that $\sign{\tu_i} = 0 \implies \sign{\tsat{\tu}_i} = 0$, and if $\sign{\tu_i} \neq 0$ then either $i \in \mathcal{I}^0$ or $\sign{\tsat{\tu}_i} = \sign{\tu_i}$. Hence the term \eqref{eq:vdot b} can be bounded by
    \begin{align}
        &- \sign{\tu}^\top H \Tilde{b}(\tsat{\tu}) \\
        \geq& \sign{\tsat{\tu}}^\top H \Tilde{b}(\tsat{\tu}) \\
        &- \sum_{j \in \mathcal{I}^0} \eta_j |\Tilde{b}_i(\tsat{\tu})|.
    \end{align}
    As Assumption \ref{ass:assumptions on b} holds, we can employ Lemma \ref{lem:sum of b} to show that this expression is strictly negative when $\tsat{\tu} \neq 0$. All together, if $\tu \neq 0$, we must have $\tsat{\tu} \neq 0$, $\tdz{\tu} \neq 0$ or both, in which case $\dot{V}(\tz,\tu)$ will be negative due to the above arguments. If $\tu = 0$, the term \eqref{eq:vdot linear} is negative definite in $\tz$. Hence $\dot{V}(\tz,\tu)$ is negative definite. As $V(\tz, \tu) > 0$ and $\dot{V}(\tz,\tu) < 0$ for any non-zero pair $(\tz, \tu)$ the equilibrium $(\zeta^0, u^0)$ is globally asymptotically stable for \eqref{eq:changed coordinates} and must therefore be unique. This can be translated to a unique equilibrium $(x^0, z^0)$ in the original coordinate frame, thus concluding the proof.
\end{proof}

To prove Theorem \ref{thm: coordinating optimality} regarding stability of the coordinating closed-loop system, we utilize the following result.
\begin{lemma} \label{lem: coordinating unsaturated}
    Assume that Assumptions \ref{ass:assumptions on b} and \ref{ass:coordinating controller tuning} hold. Furthermore, assume that $b(\overline{l}) + w > 0$ and $b(\underline{l}) + w < 0$. Then, from any initial condition, the coordinating closed-loop system \eqref{eq:coordinated closed loop} will converge to a forward invariant set in which the control signal is unsaturated, i.e., $\dz{u} = 0$.
\end{lemma}

\begin{proof}
For this proof, we will employ the coordinate frame $(\zeta, u)$ as in \eqref{eq:changed coordinates}, now with $S = \kc \ones \ones^\top$. Consider the Lyapunov function candidate
\begin{align}
    V(\zeta, u) =& \frac{1}{2}\dz{\zeta}^\top D R\inv C\inv \dz{\zeta} \nonumber\\
    &+ \frac{1}{2}\dz{u}^\top  R\inv \dz{u}.
\end{align}
Here $\dz{\zeta}$ is to be understood to element-wise use the same bounds $\overline{l}$ and $\underline{l}$ as $\dz{u}$, i.e., if $\zeta = u$ then $\dz{\zeta} = \dz{u}$. Note that 
\begin{equation}
    \frac{1}{2} \cdot \frac{d}{dt}\dz{x}^2 = \dz{x} \frac{\partial \dz{x}}{\partial x} \dot{x} = \dz{x} \dot{x} ,
\end{equation}
as $\frac{\partial \dz{x}}{\partial x} = 0$ when $\dz{x} = 0$, and $\frac{\partial \dz{x}}{\partial x} = 1$ when $\dz{x} \neq 0$. Hence we find that 
\begin{align}
     &\dot{V}(\zeta,u) \nonumber\\
     =& \dz{\zeta}^\top D R\inv C\inv \dot{\zeta} + \dz{u}^\top  R\inv \dot{x} \nonumber\\
     =&  \dz{\zeta}^\top D R\inv C\inv \left( C( u - \zeta) - \kc R \ones \ones^\top\right) \nonumber\\
     & + \dz{u}^\top  R\inv \left(  D( \zeta - u) \right. \nonumber \\
     &- \left. P\left( b(\sat{u}) + w\right) - \kc R \ones \ones^\top\right) \nonumber\\
     = & - \left( \dz{u} - \dz{\zeta}\right)^\top R\inv D \left(u - \zeta\right) \label{eq:Vdot coord linear}\\
     & - \kc \dz{\zeta}^\top  C\inv D  \ones \ones^\top \dz{u} \label{eq:Vdot coord cross}\\
     & - \kc \dz{u}^\top \ones \ones^\top \dz{u} \label{eq:Vdot coord dz}\\
     & - \dz{u}^\top R\inv P \left( b(\sat{u}) + w\right) \label{eq:Vdot coord b}
\end{align}
We will begin with the last term \eqref{eq:Vdot coord b}, which is strictly negative when $\dz{u} \neq 0$. We can see this by noting that $R\inv P$ is a positive diagonal matrix, and then identifying that for any $i \in 1,\dots,n$, if $\dzi{u} > 0$, $\sati{u} = \overline{l}_i$, and thus Assumption \ref{ass:assumptions on b} $(i)$ yields that
\begin{equation*}
    b_i(\sat{u}) + w_i \geq b_i(\overline{l}) + w_i > 0.
\end{equation*}
The opposite relation can be shown when $\dzi{u} < 0$. For the remaining terms, we first note that 
\begin{align}
    &\left( \dz{u} - \dz{\zeta}\right)^\top R\inv D \left(u - \zeta\right) \nonumber \\
    \geq &\left( \dz{u} - \dz{\zeta}\right)^\top R\inv D \left( \dz{u} - \dz{\zeta}\right),
\end{align}
which can trivially be shown by utilizing the definitions of $\sat{\cdot}$ and $\dz{\cdot}$. Furthermore, Assumption \ref{ass:coordinating controller tuning} yields that $C\inv D = \alpha I$ and $R\inv D = \alpha P\inv$. Hence \eqref{eq:Vdot coord linear}--\eqref{eq:Vdot coord dz} can be combined and upper bounded by 
\begin{align}
    &-\alpha\left( \dz{u} - \dz{\zeta}\right)^\top P\inv \left( \dz{u} - \dz{\zeta}\right)\nonumber \\
    &- \alpha \kc \dz{\zeta}^\top  \ones \ones^\top \dz{u} \nonumber\\
    &- \kc \dz{u}^\top \ones \ones^\top \dz{u}\nonumber \\
    =& -\alpha\left( \dz{u} - \dz{\zeta}\right)^\top (P\inv - \frac{\kc}{2}\ones \ones^\top) \left( \dz{u} - \dz{\zeta}\right) \nonumber\\
    & -\frac{\kc \alpha}{2} \dz{v}^\top \ones \ones^\top \dz{v} \nonumber\\
    & -\kc \left(1 + \frac{\alpha}{2}\right) \dz{u}^\top \ones \ones^\top \dz{u}.\nonumber
\end{align}
This expression is negative semi-definite under the condition that $\left(P\inv - \frac{\kc}{2}\ones \ones^\top\right) \succeq 0$. This is equivalent to 
\begin{equation}
    \begin{bmatrix}
        P\inv & \ones \\
        \ones^\top & \frac{2}{\kc}
    \end{bmatrix} \succeq 0, \label{eq:positive semidefinite condition}
\end{equation}
as $(P\inv - \frac{\kc}{2}\ones \ones^\top)$ is the Schur complement of this matrix. The condition \eqref{eq:positive semidefinite condition} is once again equivalent to 
\begin{equation}
    \frac{2}{\kc} - \ones^\top P \ones \geq 0 \iff  \frac{\kc}{2}\ones^\top P \ones = \frac{\kc}{2}\ones^\top \kp \geq 0.
\end{equation}
This holds, due to Assumption \ref{ass:coordinating controller tuning}. Hence $(P\inv - \frac{\kc}{2}\ones \ones^\top) \succeq 0$ and thus \eqref{eq:Vdot coord linear}--\eqref{eq:Vdot coord dz} is negative semi-definite. Therefore we have shown that $V(\zeta, u) > 0$ when $\dz{u} \neq 0$ and $\dot{V}(\zeta, u) < 0$ when $\dz{u} \neq 0$. Therefore the region where $\dz{u} = 0$ is globally attracting and forward invariant.
\end{proof}

\begin{proof}[Proof of Theorem \ref{thm: coordinating stability}]
    Lemma \ref{lem: coordinating unsaturated} proves that all trajectories of the closed-loop system will converge to, and remain in, the region where $\dz{u} = 0$. Here, the closed-loop systems of the coordinating and decentralized controllers are equivalent, and Assumption \ref{ass:coordinating controller tuning} implies that also Assumption \ref{ass:decentralized controller tuning} holds (disregarding the statements about the anti-windup gains, as they are inactive). Hence, we can invoke Theorem \ref{thm: decentralized stability} which applies to the decentralized closed-loop system. 
\end{proof}

\subsection{Optimality Proofs} \label{sec:optimality proofs}

We continue now to prove Theorems \ref{thm: decentralized optimality} and \ref{thm: coordinating optimality}.

\begin{proof}[Proof of Theorem \ref{thm: decentralized optimality}]
    To simplify notation, we will introduce $v^0 = \sat{u^0}$ and $v^\dagger = \sat{u^\dagger}$. Note that $v^\dagger \neq v^0$ by assumption. As both pairs $(x^0, u^0)$ and $(x^\dagger, u^\dagger)$ satisfy \eqref{eq:agent dynamics} with $\dot{x} = 0$, we can conclude that 
    \begin{equation*}
        \eta_ia_ix^\dagger_i = \eta_i(b_i(v^\dagger) + w_i) 
        = \eta_ia_ix_i^0 + \eta_i\left( b_i(v^\dagger)-b_i(v^0) \right)
    \end{equation*}
    for all $i = 1,\dots,n$. Therefore
    \begin{align}
        & \sum_{i=1}^n \eta_i a_i |x_i^\dagger| \nonumber\\
        =&  \sum_{i=1}^n \eta_i |a_ix_i^0 + b_i(v^\dagger)-b_i(v^0)| \nonumber\\
        \geq&  \sum_{i \in \mathcal{I}^\pm} \eta_ia_i |x_i^0| + \eta_i\sign{x_i^0}\left(b_i(v^\dagger)-b_i(v^0)\right) \nonumber\\
        &+  \sum_{j \in \mathcal{I}^0} \eta_j|b_j(v^\dagger)-b_j(v^0)| \label{eq:cost x index}
    \end{align}
    where $\mathcal{I}^\pm = \{i \in 1,\dots,n \, | \, x_i^0 \neq 0 \}$ and $\mathcal{I}^0 = \{i \in 1,\dots,n \, | \, x_i^0 = 0 \}$. As $(x^0, u^0)$ satisfies \eqref{eq:decentralized z} with $\dot{z} = 0$, and $\ka > 0$, we know that $\sign{x^0} = -\sign{\dz{u^0}}$. Hence we can continue to expand \eqref{eq:cost x index} to
    \begin{align}
        & \sum_{i=1}^n \eta_i a_i |x_i^0| \nonumber\\
        & +\sum_{i \in \mathcal{I}^\pm} \eta_i\sign{\dzi{u^0}}\left(b_i(v^0)-b_i(v^\dagger)\right) \nonumber\\
        & +\sum_{j \in \mathcal{I}^0}\eta_j|b_j(v^0)-b_j(v^\dagger)|.\label{eq:cost dz index}
    \end{align}
    We would here like to apply Lemma \ref{lem:sum of b}, but this requires the sets $\mathcal{J}^\pm = \{i \in 1,\dots,n \, | \, v_i^0 \neq v_i^\dagger \}$ and $\mathcal{J}^0 = \{i \in 1,\dots,n \, | \, v_i^0 = v_i^\dagger \}$. To continue, we note the following. For all $i \in \mathcal{I}^\pm$, $\dzi{u^0} \neq 0$, and hence $v_i^0 = \overline{l}_i$ or $v_i^0 = \underline{l}_i$. Therefore the following statements hold.
    \begin{align}
       & \eta_i\sign{\dzi{u^0}}\left(b_i(v^0)-b_i(v^\dagger)\right)  &&\nonumber\\
       =& \eta_i\sign{v^\dagger_i - v^0_i}\left(b_i(v^0)-b_i(v^\dagger)\right),   
       &&\forall i \in  \mathcal{I}^\pm \cap \mathcal{J}^\pm \\
        &\eta_i\sign{\dzi{u^0}}\left(b_i(v^0)-b_i(v^\dagger)\right)  &&\nonumber\\
        \geq&  -\eta_i\left|b_i(v^0)-b_i(v^\dagger)\right|,  
        &&\forall i \in  \mathcal{I}^\pm \cap \mathcal{J}^0 \\
        &\eta_i\left|b_i(v^0)-b_i(v^\dagger)\right|  &&\nonumber\\ 
        \geq& \eta_i\sign{\dzi{u^0}}\left(b_i(v^0)-b_i(v^\dagger)\right) , 
        &&\forall i \in  \mathcal{I}^0 \cap \mathcal{J}^\pm \\
        &\eta_i\left|b_i(v^0)-b_i(v^\dagger)\right|  &&\nonumber\\
        \geq& -\eta_i\left|b_i(v^0)-b_i(v^\dagger)\right| , 
        &&\forall i \in  \mathcal{I}^0 \cap \mathcal{J}^0 . 
    \end{align}
    We can use these relations to reorganize and upper-bound the sums in \eqref{eq:cost dz index}, and therefore state that
    \begin{align}
        & \sum_{i=1}^n \eta_i a_i |x_i^\dagger| \nonumber\\
        \geq & \sum_{i=1}^n \eta_i a_i |x_i^0| \nonumber\\
        & +\sum_{i \in \mathcal{J}^\pm} \eta_i\sign{v^0 -v_i^\dagger}\left(b_i(v^0)-b_i(v^\dagger)\right) \nonumber\\
        & -\sum_{j \in \mathcal{J}^0}\eta_j|b_j(v^0)-b_j(v^\dagger)| \nonumber \\
        > & \sum_{i=1}^n \eta_i a_i |x_i^0|.
    \end{align}
    The final inequality derives from Lemma \ref{lem:sum of b}, as we assume that Assumption \ref{ass:assumptions on b} holds and $v^0 \neq v^\dagger$. 
\end{proof}

\begin{proof}[Proof of Theorem \ref{thm: coordinating optimality}]
    Since both $(x^0, u^0)$ and $(x^\dagger, u^\dagger)$ solve \eqref{eq:agent dynamics} with $\dot{x} = 0$, we know that
    \begin{align}
        x^\dagger =& A\inv \left( b(\sat{u^\dagger}) + w \right) \nonumber \\
        =& x^0 +  A\inv \left( b(\sat{u^\dagger})  -  b(\sat{u^0})\right), \label{eq:xdagger expression}
    \end{align}
    where $A = \diagm{a}$. We also know that
    \begin{equation}
        x^0 = - \kc \ones \ones^\top \dz{u^0}
        \label{eq:steady state x0}
    \end{equation}
    because $(x^0, u^0)$ is an equilibrium for the closed-loop system \eqref{eq:coordinated closed loop} and thus satisfy \eqref{eq:coordinated z} with $\dot{z} = 0$. We therefore know that $x^0$ is proportional to the vector $\ones$, and thus $\norm{\infty}{x^0} = \max_i |x_i^0|$ is maximized by all $i = 1,\dots,n$ simultaneously. Consider first the case where $x^0 > 0$. Then the contradictory notion that $\norm{\infty}{x^\dagger} \leq \norm{\infty}{x^0}$ would therefore require that $x^\dagger \leq x^0$, and by \eqref{eq:xdagger expression} also $ b(\sat{u^\dagger})  -  b(\sat{u^0}) \leq 0$. This is however impossible, because under Assumption \ref{ass:assumptions on b}, Lemma \ref{lem:b inverse positive} states that $b(\sat{u^\dagger})  -  b(\sat{u^0})  \leq 0$ requires $\sat{u^\dagger} < \sat{u^0}$. This is incompatible with \eqref{eq:steady state x0}, stating that $\ones^\top \dz{u^0} < 0$, i.e., there must exist at least one $i \in 1,\dots,n$ such that $\sati{u^0} = \underline{l}_i$, which means that necessarily $\sati{u^\dagger} \geq \sati{u^0}$, establishing a contradiction. We can make a symmetric argument to discard the possibility that $x^0 < 0$ and $x^\dagger \geq x^0$. Thus the only remaining option is $x^0 = x^\dagger = 0$, which requires $b(\sat{u^\dagger})  =  b(\sat{u^0})$, and thus $u^\dagger = u^0$, contradicting the assumption of the theorem statement. This concludes the proof.
\end{proof}

\section{Conclusion} \label{sec:conclusion}
In this paper we considered a particular class of multi-agent systems, where the agents are connected through a capacity-constrained nonlinearity. For this type of system, we considered two proportional-integral controllers equipped with anti-windup compensation: One which was fully decentralized, and one in which the anti-windup compensator introduces a rank-1 coordinating term. We showed that the equilibria of these two closed-loop system were optimal, in the sense that they minimized the size of the control errors $x$ in terms of the costs $\sum a_i \eta_i |x_i|$ and $\norm{\infty}{x}$ respectively. Additionally, we showed that the fully decentralized strategy provides guarantees of global, asymptotic stability with regards to a unique equilibrium. For the coordinating controller, we demonstrated global asymptotic stability when the disturbance can be rejected.

To demonstrate the applicability of the considered model, we showed how it can capture the indoor temperatures of buildings connected through a district heating network. In this setting, the capacity-constrained nonlinear interconnection consists of the hydraulics mapping the valve positions of each building to the resulting flow rates in the system. We demonstrated in a numerical example how the two considered controllers could then achieve different design goals - minimizing the average or the worst-case temperature deviations in the system respectively.

There are plenty of outlooks for future work: The internal agent dynamics are currently simple and on the form $-a_i x_i$. This could perhaps be extended to more complex dynamics, where some stability assumptions are placed on the dynamics of each agent.
Another outlook is analyzing the transient behavior of these systems. This would include understanding the effect of slowly time-varying $w(t)$ and $b(v, t)$. In the district heating setting which we considered in this paper, this would account for changes in outdoor temperature and changes in the supply-temperature in the network. 
The cost-functions which are asymptotically minimized by the considered controllers could perhaps be generalized. Developments in this direction include design of controllers which maintain scalability and structure when considering other cost functions, as well as quantifying the suboptimality attained in utilizing one of the considered controllers in this paper for other cost functions.
Finally, stronger stability guarantees can likely be established for the coordinating controller, which are applicable even when the stabilized equilibrium lies in the saturated domain.

\section{Acknowledgments}                              
This work is funded by the European Research Council (ERC) under the European Union's Horizon 2020 research and innovation program under grant agreement No 834142 (ScalableControl). 

The authors are members of the ELLIIT Strategic Research Area at Lund University.

\bibliographystyle{ieeetr}
\bibliography{bibliography}

\section*{Appendix: Remaining Proofs}
\begin{proof}[Proof of Lemma \ref{lem:sum of b}.]
    Introduce the difference between the inputs $\mu$ as $\mu = \Tilde{v} - v$. 
    We will then split the set $\mathcal{I}^\pm$ into the two sets $\mathcal{I}^+ = \{i \in 1,\dots,n \,| \, \sign{\mu_i} = 1 \}$ and $\mathcal{I}^- = \{i \in 1,\dots,n \,| \, \sign{\mu_i} = -1 \}$. For all $i \in \mathcal{I}^+$, we can invoke Assumption \ref{ass:assumptions on b} $(i)$ to state that 
    \begin{align}
        b_i(v + \mu) \geq b_i(v + \pospart{\mu}), 
        b_i(v) \leq b_i(v + \negpart{\mu}).
    \end{align}
    Conversely, for any $j \in \mathcal{I}^-$,
    \begin{align}
        b_j(v + \mu) \leq b_j(v + \negpart{\mu}), 
        b_j(v) \geq b_j(v + \pospart{\mu}).
    \end{align}
    Finally for all $k \in \mathcal{I}^0$, Assumption \ref{ass:assumptions on b} $(i)$ implies 
    \begin{align}
        |b_k(v + \mu) - b_k(v)| \geq& b_k(v + \mu) - b_k(v) \nonumber \\
        \geq& b_k(v + \pospart{\mu}) - b_k(v + \negpart{\mu}).
    \end{align}
    Hence can expand \eqref{eq:sum of b with signs} as
    \begin{align}
        &\sum_{i \in \mathcal{I}^+} \eta_i\left(b_i(v + \mu) -b_i(v)\right) \nonumber \\
        &+ \sum_{j \in \mathcal{I}^-} \eta_j\left(b_j(v) -b_j(v + \mu)\right)  \\
        &- \sum_{k \in \mathcal{I}^0} \eta_k|b_k(v + \mu) -b_k(v)| \nonumber\\
        \geq& \sum_i \eta_i\left(b_i(v + \pospart{\mu}) -b_i(v + \negpart{\mu})\right) \\
        =& \eta^\top \left( b(v + \pospart{\mu}) - b(v + \negpart{\mu}) \right).
    \end{align}
    By Assumption \ref{ass:assumptions on b} $(ii)$, this quantity is strictly positive when $\pospart{\mu} - \negpart{\mu} \geq 0$, which holds for any $\mu \neq 0$, thus concluding the proof.
     
\end{proof}

\begin{proof}[Proof of Lemma \ref{lem:b inverse positive}]
    Let $\mu = \overline{v} - \underline{v}$. Consider first the contradictory notion that $\negpart{\mu} \neq 0$. Then Assumption \ref{ass:assumptions on b} $(ii)$ yields that
    \begin{equation}
        \eta^\top \left( b(\negpart{\mu} + \underline{v}) - b(v) \right) < 0.
    \end{equation}
    Hence $\exists k \in 1,\dots,n$ such that $b_k(\negpart{\mu} + \underline{v}) - b_k(\underline{v})$, and thus due to \ref{ass:assumptions on b} $(i)$ we can conclude that $\mu_k < 0$. But since $\mu_k < 0$, we can also use Assumption \ref{ass:assumptions on b} $(i)$ to find
    \begin{equation}
        b_k(\mu + \underline{v}) - b_k(\underline{v}) \geq b_k(\negpart{\mu} + \underline{v}) - b_k(\underline{v}) < 0.
    \end{equation}
    This contradicts the assumption of the Lemma, and hence we conclude that $\negpart{\mu} = 0$. Therefore $\mu \geq 0$. Now consider the notion that $\exists$ $j$ such that $\mu_j = 0$. Then if $\mu \neq 0$, Assumption \ref{ass:assumptions on b} $(i)$ states that $b_j(\mu + \underline{v}) - b_j(\underline{v}) < 0$, which contradicts the lemma assumption. Hence either $\mu > 0$ or $\mu = 0$. Clearly $\mu = 0$ is impossible, as this would contradict the assumption of the lemma. Thus $\mu = \overline{v} - \underline{v} > 0$.
     
\end{proof}

\begin{proof}[Proof sketch of Lemma \ref{lem:decentralized equilibrium}.]
    An equilibrium $(x^0,z^0)$ with associated stationary control input $u^0$ satisfies \eqref{eq:decentralized closed loop} with $\dot{x} = \dot{z} = 0$. Note that to solve this system of equations, it is sufficient to find $u^0$ which satisfies
    \begin{equation*}
        0 = b_i(\sat{u^0}) + w_i +a_i \kai \dz{u_i^0}, \, \forall i = 1,\dots,n. \label{eq:fixed point}
    \end{equation*}
    This $u^0$ uniquely fixes $x^0$ through \eqref{eq:decentralized x} and $z^0$ through \eqref{eq:decentralized u}. Equation \eqref{eq:fixed point} equates to solving a fixed-point problem $u^0 = T(u^0)$, where we define the map $T$ as
    \begin{equation*}
        T_i(u^0) = u_i^0 - \alpha\left(b_i(\sat{u^0}) + w_i + a_i \kai \dz{u_i^0} \right) 
    \end{equation*}
    where $\alpha > 0$ is chosen sufficiently small. $T$ is clearly forward-invariant with respect to a sufficiently large box $\mathcal{C}$ of size $c$, $\mathcal{C} = \{ u \in \R^n\, | \, \norm{\infty}{u} \leq c \}$, which allows us to invoke Brouwer's fixed-point theorem.
\end{proof}

\end{document}